\documentclass[11pt,draftcls]{IEEEtran}{\onecolumn}
\usepackage{amsthm}
\usepackage{amsmath}
\usepackage{slashbox}
\usepackage{graphicx}
\usepackage{epstopdf}
\usepackage{epsfig,graphics,subfigure,psfrag,amsmath,amssymb}
\usepackage[nooneline,flushleft]{caption2}
\usepackage{cite}
\usepackage{float}
\usepackage{tocvsec2}
\usepackage{color}
\usepackage{verbatim}
\usepackage{algorithm, algorithmic}
\usepackage{mathtools}
\usepackage{hhline}
\usepackage{enumerate}

\newtheorem{thm}{Theorem}

\newtheorem{rem}{Remark}

\makeatletter

\newcommand{\Rmnum}[1]{\expandafter\@slowromancap\romannumeral #1@}
\makeatother

\begin{document}
\title{A Graphical Evolutionary Game Approach to Social Learning}
\author{Xuanyu Cao and K. J. Ray Liu, \emph{Fellow, IEEE}}
\maketitle
\begin{abstract}
In this work, we study the social learning problem, in which agents of a networked system collaborate to detect the state of the nature based on their private signals. A novel distributed graphical evolutionary game theoretic learning method is proposed. In the proposed game-theoretic method, agents only need to communicate their binary decisions rather than the real-valued beliefs with their neighbors, which endows the method with low communication complexity. Under mean field approximations, we theoretically analyze the steady state equilibria of the game and show that the evolutionarily stable states (ESSs) coincide with the decisions of the benchmark centralized detector. Numerical experiments are implemented to confirm the effectiveness of the proposed game-theoretic learning method.\footnote{Copyright (c) 2017 IEEE. Personal use of this material is permitted. However, permission to use this material for any other purposes must be obtained from the IEEE by sending a request to pubs-permissions@ieee.org.}
\end{abstract}


\section{Introduction}

In the recent decade, tremendous research efforts have been devoted to the social learning problems, in which agents of networked systems learn from not only their own private signals but also other agents. Applications of such social learning problems are ubiquitous in fields such as state estimation in power systems, distributed detection over sensor networks and behavior analysis of social networks.

The setups of the social learning problems can be sorted into two categories. In the first category, agents arrive at the system sequentially and make one-shot decisions consecutively based on their own private observations and the actions of their predecessors. In \cite{krishnamurthy2011bayesian,krishnamurthy2012quickest}, the sequential detection problem was studied by using partially observable Markov decision processes (POMDP). The impact of the memory size of the agents was investigated in \cite{koplowitz1975necessary,drakopoulos2013learning}. The effect of noisy communications in sequential detection was considered in \cite{zhang2013hypothesis} while the benefits of randomness of decision making were studied in \cite{wang2015social}. Furthermore, a Chinese restaurant game-theoretic analysis of agents' sequential decision making processes was presented in \cite{wang2013sequential}. The problem formulation of this letter is closer to the second category of social learning setups, in which agents are fixed and networked and update their actions iteratively based on their own private signals and neighbor agents' actions or beliefs. In this line of models, the consensus hypothesis testing over networks was investigated in \cite{braca2010asymptotic,wang2013gossip}, while its applications in wireless communications were considered in \cite{mostofi2010binary}. Moreover, a non-Bayesian social learning method with the property of asymptotic learning was proposed in \cite{jadbabaie2012non}. Overviews on topics in social learning, distributed detection/estimation over networks were presented in \cite{djuric2012distributed,dimakis2010gossip,olfati2007consensus}. 
 
As the social learning problems involve interactions (learning and decision making) between multiple agents, game theory emerges as an appropriate mathematical tool to study them \cite{tadelis2013game}. In \cite{eksin2014bayesian,eksin2013learning}, a Bayesian quadratic network game filter was proposed for rational agents to learn the state of the world in a cooperative manner. A Bayesian dynamic game model of social learning was investigated and the conditions of asymptotic learning were presented in \cite{acemoglu2011bayesian}. Additionally, the network users' decision making problems were studied with a Bayesian game formulation in \cite{krishnamurthy2013social,krishnamurthy2014tutorial}.

In this work, inspired by the recent success of evolutionary game theory in diverse fields \cite{jiang2013distributed,jiang2014graphical,jiang2014evolutionary,cao2016evolutionary,nowak2004evolutionary,ohtsuki2007breaking}, we propose a graphical evolutionary game-theoretic method for social learning. In the proposed method, based on a death-birth decision update rule, agents only need to communicate their binary decisions instead of the real-valued beliefs with their neighbors, which endows the method with low communication complexity. By invoking mean field approximations, we analyze the steady state equilibria of the game and show that the evolutionarily stable states (ESSs) \cite{weibull1997evolutionary} coincide with the decisions of the benchmark centralized detector. Lastly, we present numerical results to confirm the effectiveness of the proposed game-theoretic learning method.

\section{Problem Formulation}
Consider a network of $N$ agents or nodes (the two terms are used interchangeably in the following). Assume for simplicity that the network is $k$-regular, i.e., the degree (number of neighbors) of each agent is $k$. In practice, many networks are $k$-regular graphs. For example, many sensor networks are grid networks over 2-dimensional plane and are thus 4-regular graphs \cite{karamchandani2011time}; many cellular communication networks are comprised of hexagon cells (each hexagon cell corresponds to the service area of one base station) and are hence 6-regular graphs. In the social learning problem, there is an unknown state of the nature $\theta\in\{0,1\}$ to be detected by all the nodes in a collaborative manner based on their individual private signals or measurements. Suppose the prior distribution of $\theta$ is $\Pr(\theta=0)=\Pr(\theta=1)=0.5$. Agents are sorted into $I$ categories depending on the qualities of their private signals, i.e., the usefulness of the private signals in detecting the unknown state $\theta$. Suppose agent $n$ has some private signal $s_n$ and its type is $i$. Then, its private belief is $p_i=\Pr\{\theta=0|s_n\}$. Clearly, if $p_i$ is close to 0 or 1, then the signals of type-$i$ agents are useful for detecting $\theta$. Oppositely, if $p_i$ is close to $0.5$, then the signals of type-$i$ agents are not very useful.

\subsection{The Centralized Detector}

In this subsection, a centralized detector, i.e., a detector utilizing the signals of all agents in a centralized manner, is derived as a performance benchmark. Assume that, given the true state $\theta$, the signals $s_1,...,s_N$ (henceforth $s_{1:N}$ for shorthand) are conditionally independent, i.e., $p(s_{1:N}|\theta)=\Pi_{n=1}^Np(s_n|\theta)$. With the signals $s_{1:N}$ of all $N$ nodes, a centralized processor can form the posterior distribution $\Pr(\theta=0|s_{1:N})$ according to the Bayesian rule as follows:
\begin{align}
&\Pr(\theta=0|s_{1:N})\\
&=\frac{p(s_{1:N}|\theta=0)\Pr(\theta=0)}{p(s_{1:N}|\theta=0)\Pr(\theta=0)+p(s_{1:N}|\theta=1)\Pr(\theta=1)}\\
&=\frac{\Pi_{n=1}^Np(s_n|\theta=0)}{\Pi_{n=1}^Np(s_n|\theta=0)+\Pi_{n=1}^Np(s_n|\theta=1)}\\
&=\frac{1}{1+\Pi_{i=1}^I\left(\frac{1-p_i}{p_i}\right)^{Nq_i}},
\end{align}
where we denote the proportion of type-$i$ agents as $q_i$. With a threshold of $0.5$ for the posterior distribution $\Pr(\theta=0|s_{1:N})$, the decision rule of the centralized detector $\hat{\theta}_c$ is given by:
\begin{align}\label{centralized}
\sum_{i=1}^Iq_i\log\left(\frac{1-p_i}{p_i}\right)\underset{\hat{\theta}_c=0}{\overset{\hat{\theta}_c=1}{\gtrless}}0.
\end{align}

\subsection{A Graphical Evolutionary Game Framework}
The centralized detector has several drawbacks such as large communication overhead and vulnerability to link failures which make it infeasible in many applications. Therefore, we are motivated to find another detection algorithm with the following favorable properties.
\begin{itemize}
\item \texttt{P1} The detection algorithm is distributed, i.e., each agent only communicates with its neighbors and no centralized entity is needed.
\item \texttt{P2} Agents only interchange their current binary decisions on the state $\theta$ instead of their real-valued beliefs (posterior distributions) on $\theta$. This reduces the communication complexity significantly.
\item \texttt{P3} The detection algorithm produces the same result as the centralized detector \eqref{centralized} does, possibly asymptotically if the algorithm is iterative.
\end{itemize}
In this subsection, we present a graphical evolutionary game social learning approach, which satisfies the aforementioned three properties. Suppose each agent $n$ has a decision $d_n\in\{0,1\}$ on the state $\theta$ and the decision $d_n$ shall be updated iteratively in the game. When a type-$i$ agent $n$ interacts with one of its neighbors, agent $m$, the utility of agent $n$ is summarized in the following table for different combinations of actions of the two interacting parties $n$ and $m$:
\begin{center}
  \begin{tabular}{ c | c | c }

    & $d_m=0$ & $d_m=1$ \\ \hline
    $d_n=0$ & $\log(1-p_i)+u$ & $-\log(1-p_i)$ \\ \hline
    $d_n=1$ & $-\log p_i$ & $-\log p_i+u$ \\

  \end{tabular}
\end{center}
Here $u\geq0$ is some non-negative constant used to capture the fact that agents tend to imitate their neighbors (or friends in social networks) and reach consensus. Additionally, agent $n$ also tends to adhere to its own private belief $p_i$. As such, we reward or penalize the utility of agent $n$ for actions conforming to or deviating from its belief $p_i$, respectively. The usage of logarithmic terms in the utility is inspired by the centralized detector \eqref{centralized}. For an agent with total utility $U$ through interactions with her neighbors, we further define her fitness $\pi$ as a convex combination of $U$ and 1: $\pi=1-\alpha+\alpha U$, where $\alpha>0$ is some small positive constant called the selection strength in evolutionary game theory. The bigger the selection strength $\alpha$, the more heavily the fitness $\pi$ depends on the utility $U$ and the bigger the advantage of agents with large utility. For a type-$i$ agent with $k_0$ neighbors making decision $0$, if he makes decision $d=0$, then his fitness is:
\begin{align}
\begin{split}\label{pi0}
\pi_0(i,k_0)=&1-\alpha+\alpha[k_0(-\log(1-p_i)+u)-(k-k_0)\log(1-p_i)].
\end{split}
\end{align}
If he makes decision $d=1$, then his fitness is:
\begin{align}\label{pi1}
\pi_1(i,k_0)=1-\alpha+\alpha[-k_0\log p_i+(k-k_0)(-\log p_i+u)].
\end{align}

Based on fitness, agents can update their decisions according to some strategy update rule. In the literature of graphical evolutionary game theory \cite{ohtsuki2007breaking,ohtsuki2006replicator,jiang2014graphical,jiang2014evolutionary,cao2016evolutionary}, there are mainly three strategy update rules: the death-birth process, the birth-death process and the imitation process. In this letter, we will focus on the death-birth update rule and other rules can be similarly analyzed. In the death-birth update rule, at each time slot, one agent is selected to abandon her decision uniformly randomly (death process) and the chosen agent update her decision to be one of her neighbors' decisions with probability proportional to their fitness (birth process). This decision update process continues repeatedly across time. In this work, our goal is to study the agents' steady state behaviors in this update process.

The proportion of adoption of decision $0$ among type-$i$ agents is denoted as $x_i$ while the proportion of adoption of decision $0$ among all agents is denoted as $x$. We call $x_i$ the population dynamics of type-$i$ agents and $x$ the population dynamics of the entire network (or simply population dynamics). Obviously, we have $x=\sum_{i=1}^Iq_ix_i$. Our goal in this letter is to study the steady state equilibrium of the population dynamics $x$ and show that this equilibrium coincides with the centralized detector \eqref{centralized}.

We note that the gossip method proposed in \cite{wang2013gossip} tackles a similar social learning problem and also possesses properties \texttt{P1}, \texttt{P2} and \texttt{P3}. In this paper, we take an alternative approach based on evolutionary game theory as opposed to the gossip based method in \cite{wang2013gossip}. The proposed game-theoretic social learning method takes agents' rational learning and decision-making behaviors (such as learning from neighbors with high fitness) into consideration and is thus more amenable to practical implementations in systems with intellegient or strategic agents, e.g., social networks.

\section{Algorithm Development and Equilibrium Analysis}

In this section, we develop the detailed algorithm of the game-theoretic social learning method and analyze the corresponding steady state equilibrium, i.e., the evolutionarily stable state (ESS) \cite{weibull1997evolutionary}, of the population dynamics $x$. Suppose, at a time instant, a type-$i$ agent with decision $0$ is chosen to abandon her decision. According to the death-birth update rule, this agent should update her decision to be one of her neighbors' decisions with probability proportional to fitness. However, as we only allow the agents to communicate their decisions $d$ rather than their private beliefs $p$ (property \texttt{P2}), the chosen agent is unaware of her neighbors' fitness, which depend on their private beliefs. As such, the chosen agent will update her decision as if all of her neighbors' types are $i$, i.e., their beliefs are $p_i$, and only take the neighbors' decisions into consideration. Thus, the probability that the chosen agent will change her decision from 0 to 1 is given by:
\begin{align}
\Pr_{0\rightarrow1}(i,k_0)=\frac{(k-k_0)\pi_1(i,k_0)}{k_0\pi_0(i,k_0)+(k-k_0)\pi_1(i,k_0)}.
\end{align}
Exploiting the expressions of fitness in \eqref{pi0} and \eqref{pi1} and making use of the first order approximation $\frac{1+a\alpha}{1+b\alpha}\approx 1+(a-b)\alpha$ for small $\alpha$, we compute the transition probability $\Pr_{0\rightarrow1}(i,k_0)$ in \eqref{trans}.
\begin{footnotesize}
\begin{align}
&\Pr_{0\rightarrow 1}(i,k_0)\nonumber\\
&=\frac{k-k_0}{k}\frac{1+\alpha[-k_0\log p_i+(k-k_0)(-\log p_i+u)-1]}{1+\alpha\left\{\frac{k_0}{k}[k_0(-\log(1-p_i)+u)-(k-k_0)\log(1-p_i)-1]+\left(1-\frac{k_0}{k}\right)[-k_0\log p_i+(k-k_0)(-\log p_i+u)-1]\right\}}\nonumber\\
%&\approx\frac{k-k_0}{k}\left\{1+\alpha\left(\frac{k_0}{k}[-k_0\log p_i+(k-k_0)(-\log p_i+u)-1]-\frac{k_0}{k}[k_0(-\log(1-p_i)+u)-(k-k_0)\log(1-p_i)-1]\right)\right\}\nonumber\\
&\approx\frac{k-k_0}{k}+\alpha(k-k_0)\left[\left(\log\frac{1-p_i}{p_i}+u\right)\frac{k_0}{k}-2u\frac{k_0^2}{k^2}\right]\label{trans}
\end{align}
\end{footnotesize}

Note that $k_0$ is a binomially distributed random variable with probability mass function (PMF) $
\beta(k,k_0)=
\left(
\begin{aligned}
&k\\
&k_0
\end{aligned}
\right)x^{k_0}(1-x)^{k-k_0}.
$
Using the moments of binomial distribution, we obtain $\mathbb{E}[k_0]=kx$, $\mathbb{E}[k_0^2]=(k^2-k)x^2+kx$, $\mathbb{E}[k_0^3]=k(k-1)(k-2)x^3+3k(k-1)x^2+kx$. Thus, we can compute the expected transition probability averaged over $k_0$:
\begin{align}\label{3}
&\mathbb{E}_{k_0}\left[\Pr_{0\rightarrow1}(i,k_0)\right]\nonumber\\
&=1-x+\alpha\left(\log\frac{1-p_i}{p_i}+u\right)[-(k-1)x^2+(k-1)x]\nonumber\\
&~~~-2u\alpha[(-k+3-2k^{-1})x^3+(k-4+3k^{-1})x^2\nonumber\\
&~~~+(1-k^{-1})x]
\end{align}
Noticing that the probability of choosing a type-$i$ agent with decision $0$ to abandon her decision is $q_ix_i$, we write the PMF of the increment of $x_i$, denoted as $\delta x_i$, in the following:
\begin{align}\label{1}
\Pr\left(\delta x_i=-\frac{1}{Nq_i}\right)=q_ix_i\mathbb{E}\left[\Pr_{0\rightarrow1}(i,k_0)\right].
\end{align}
Similarly, by considering the scenario where a type-$i$ agent with decision $1$ is selected to abandon her decision, we get:
\begin{align}\label{2}
\Pr\left(\delta x_i=\frac{1}{Nq_i}\right)&=q_i(1-x_i)\mathbb{E}\left[\Pr_{1\rightarrow0}(i,k_0)\right]\\
&=q_i(1-x_i)\left(1-\mathbb{E}\left[\Pr_{0\rightarrow1}(i,k_0)\right]\right).
\end{align}
We approximate the discrete time decision update system with a continuous time version, as per convention in the analysis of graphical evolutionary game \cite{ohtsuki2007breaking,ohtsuki2006replicator,jiang2014graphical,jiang2014evolutionary,cao2016evolutionary}. Thus, utilizing \eqref{3}, \eqref{1} and \eqref{2}, we derive the evolutionary dynamics of $x_i$ as follows:
\begin{align}
&\dot{x}_i=\frac{1}{Nq_i}\Pr\left(\delta x_i=\frac{1}{Nq_i}\right)-\frac{1}{Nq_i}\Pr\left(\delta x_i=-\frac{1}{Nq_i}\right)\nonumber\\
&=\frac{1}{N}\left(1-x_i-\mathbb{E}\left[\Pr_{0\rightarrow1}(i,k_0)\right]\right)\nonumber\\
&=\frac{x}{N}-\frac{x_i}{N}+\frac{\alpha}{N}x(x-1)\bigg\{2u\big[\left(-k+3-2k^{-1}\right)x\nonumber\\
&~~~-1+k^{-1}\big]+\left(\log\frac{1-p_i}{p_i}+u\right)(k-1)\bigg\}\label{x_i_dynamics}
\end{align}
Taking a weighted average over all types, we get the evolutionary dynamics of the population dynamics $x$ as:
\begin{align}\label{x_dynamics}
\dot{x}\nonumber&=\sum_{i=1}^Iq_i\dot{x}_i\\
&=\frac{\alpha}{N}x(x-1)\big\{2u\left[\left(-k+3-2k^{-1}\right)x-1+k^{-1}\right]\nonumber\\
&~~~+(\lambda+u)(k-1)\big\},
\end{align}
where $\lambda\triangleq\sum_{i=1}^Iq_i\log\frac{1-p_i}{p_i}$. Note that $\lambda$ is just the discriminant used in the centralized detector \eqref{centralized}. Now, we are ready to present the main theorem of this letter regarding the ESS of the social learning game.

\begin{thm}
(i) Suppose the degree $k\geq2$. Then, the set of evolutionarily stable states (ESSs) $\mathcal{X}^*$ of the social learning game is:
\begin{align}\nonumber
\mathcal{X}^*=\begin{cases}
\{0\},~~\text{if}~~~\lambda>u-2k^{-1}u,\\
\{1\},~~\text{if}~~~\lambda<-u+2k^{-1}u,\\
\{0,1\},~~\text{if}~~~-u+2k^{-1}u<\lambda<u-2k^{-1}u.
\end{cases}
\end{align}

(ii) If we further assume that the initial value of the population dynamics $x$ is $x(0)=0.5$, which can be achieved by a random guess by all agents, then the ESS $x^*$ that the population dynamics $x$ converges to is:
\begin{align}\label{game_detector}
x^*=\begin{cases}
0,~~\text{if}~~~\lambda>0,\\
1,~~\text{if}~~~\lambda<0.
\end{cases}
\end{align}
\end{thm}
\begin{proof}
(i) Letting $\dot{x}=0$ in the population dynamics \eqref{x_dynamics} yields three equilibria $0$, $1$, and $\tilde{x}$, where $\tilde{x}\triangleq\frac{\lambda}{2u\left(1-2k^{-1}\right)}+\frac{1}{2}$.
For an equilibrium point to be an ESS, it needs to be a locally asymptotically stable for the underlying dynamical system. To test the stability of the three equilibria, we form the Jacobian matrix $\mathbf{J}\in\mathbb{R}^{2\times2}$ of the dynamical system $(x_i,x)$ specified in equations \eqref{x_i_dynamics} and \eqref{x_dynamics}:
\begin{align}
\mathbf{J}=\left[
\begin{array}{cc}
\frac{\partial\dot{x}_i}{\partial x_i} & \frac{\partial\dot{x}_i}{\partial x}\\
\frac{\partial\dot{x}}{\partial x_i} & \frac{\partial\dot{x}}{\partial x}
\end{array}
\right].
\end{align}
The entries of $\mathbf{J}$ are computed as follows:
\begin{align}\nonumber
\frac{\partial\dot{x}_i}{\partial x_i}=-\frac{1}{N},~~~\frac{\partial\dot{x}}{\partial x_i}=0,
\end{align}
\begin{align}
&\frac{\partial\dot{x}_i}{\partial x}=\frac{1}{N}+\frac{\alpha}{N}(2x-1)\bigg[-2u\big(\big(k-3+2k^{-1}\big)x+1-k^{-1}\big)\nonumber\\
&+\left(\log\frac{1-p_i}{p_i}+u\right)(k-1)\bigg]+\frac{2u\alpha}{N}x(x-1)\left(-k+3-2k^{-1}\right),\nonumber
\end{align}
\begin{align}
\frac{\partial\dot{x}}{\partial x}&=\frac{\alpha}{N}(2x-1)\big[-2u\big(\left(k-3+2k^{-1}\right)x+1-k^{-1}\big)\nonumber\\
&~~+(u+\lambda)(k-1)\big]+\frac{2u\alpha}{N}x(x-1)(-k+3-2k^{-1})\nonumber
\end{align}
As $\mathbf{J}$ is upper triangular and $\frac{\partial\dot{x}_i}{\partial x_i}$ is negative, the locally asymptotically stability is equivalent to $\frac{\partial\dot{x}}{\partial x}<0$. Therefore, $x=0$ is an ESS iff $\frac{\partial\dot{x}}{\partial x}|_{x=0}<0$, i.e., $\lambda>-u+2k^{-1}u$. Similarly, $x=1$ is an ESS iff $\frac{\partial\dot{x}}{\partial x}|_{x=1}<0$, i.e., $\lambda<u-2k^{-1}u$. $x=\tilde{x}$ is an ESS iff $\frac{\partial\dot{x}}{\partial x}|_{x=\tilde{x}}<0$, i.e., $\tilde{x}<0$ or $\tilde{x}>1$, which contradict to the fact that the population dynamics is within $[0,1]$. So, $\tilde{x}$ can never be an ESS. We thus conclude the first part of the theorem.

(ii) If $\lambda>u-2k^{-1}u$, then the unique ESS is 0 and the population dynamics $x$ will converge to it. Similarly, if $\lambda<-u+2k^{-1}u$, then the unique ESS is 1 and the population dynamics $x$ will converge to it. In these two circumstances, \eqref{game_detector} evidently holds. If $-u+2k^{-1}u<\lambda<u-2k^{-1}u$, then the set of ESSs $\mathcal{X}^*$ contains both 0 and 1 and we need to ascertain which ESS will the population dynamics $x$ converge to. Recall the evolutionary dynamics of $x$ in \eqref{x_dynamics} and we note that if $x>\tilde{x}$, then $\dot{x}>0$ and $x$ is increasing; if $x<\tilde{x}$, then $\dot{x}<0$ and $x$ is decreasing. Recall that the initial value of $x$ is $x(0)=0.5$. If $\lambda>0$, then $\tilde{x}>0.5=x(0)$. So, $x$ is decreasing initially, which means $x$ will become even smaller and $x<\tilde{x}$ still hods. Therefore, $x$ is always decreasing and the ESS it converges to is 0. Analogously, if $\lambda<0$, then the ESS $x$ converges to is 1.
\end{proof}
\begin{rem}
Part (ii) of Theorem 1 establishes that the steady state of the game-theoretic social learning method coincides with the decision of the centralized detector, i.e., the game-theoretic social learning method possesses property \texttt{P3}.
\end{rem}

\section{Numerical Results}

In this section, numerical results are presented to corroborate the proposed game-theoretic social learning approach. We simulate a random regular network with $N=1000$ nodes (agents) and the degree of each node is $k=20$. The game parameters are chosen to be $\alpha=0.05$ and $u=0.5$. All experimental results are averages over 100 independent trials.

We first consider a network of $I=2$ types of agents. The belief of the first type is fixed to be $p_1=0.2$. We consider two scenarios (i) $q_1=q_2=0.5$; (ii) $q_1=0.3,~q_2=0.7$. The relation between the ESS and $p_2$ is reported in Fig. \ref{simu}-(a) for the two scenarios, respectively. , The ESSs are computed as the average proportion of agents with decision 0 over the 100 trials. The decisions of the centralized detector are also plotted as a benchmark. We observe that the ESSs of the game-theoretic learning are close to the decisions of the centralized detector in both scenarios. The gaps between the ESSs of the game-theoretic learning method and the decisions of the centralized detector are consequences of the randomness of the graphical evolutionary game formulation (e.g., the birth process in the death-birth decision update rule is subject to randomness). Note that the theoretical result (e.g., Theorem 1) is based on mean-field approximations, i.e., replacing random variables with their expectations to simplify analysis. Therefore, though Theorem 1 asserts that the steady states of the game-theoretic learning method coincide with the decisions of the centralized detector, there exist some gaps between the two in numerical experiments. Since the game-theoretic learning method is fully distributed and only requires communications of agents' binary decisions instead of their real-valued beliefs, it is still more desirable in many applications, especially those in need of low communication overhead and robustness.

%\begin{figure}
  %\centering
  %\includegraphics[scale=.1]{figs/2_type_ESS.eps}\\
  %\caption{Networks with two types of agents}\label{2type}
%\end{figure}

We further conduct experiments for networks with $I=5$ types of agents. The beliefs of the first four types are set to be $p_1=0.6,~p_2=0.7,~p_3=0.5,~p_4=0.4$. We consider two scenarios: (i) $q_1=q_2=q_3=q_4=q_5=0.2$; (ii) $q_1=0.2,~q_2=0.1,~q_3=0.1,~q_4=0.1,~q_5=0.5$. The relation between the ESSs and $p_5$ is illustrated in Fig. \ref{simu}-(b). The decisions of the centralized detector are also shown as a comparison. Similar to the experiments with 2 types of agents, the ESSs of the game-theoretic learning method can still match the decisions of the centralized detector approximately, which confirms the effectiveness of the proposed game-theoretic social learning method for different numbers of types.

The typical number of iterations (or equivalently, time slots) needed to converge to the ESS is between $5\times10^4$ to $10^5$. Though the iteration number seems huge, the actual convergence time in real networks is not large given the fact that the length of each time slot is very small (the length of time slot is approximately inversely proportional to the number of agents $N$ since one agent is chosen to update her decision in each time slot).

%\begin{figure}
  %\centering
  %\includegraphics[scale=.1]{figs/5_type_ESS.eps}\\
  %\caption{Networks with five types of agents}\label{5type}
%\end{figure}

\begin{figure}
\renewcommand\figurename{\small Fig.}
\centering \vspace*{8pt} \setlength{\baselineskip}{10pt}

\subfigure[Three types of agents]{
\includegraphics[scale = 0.15]{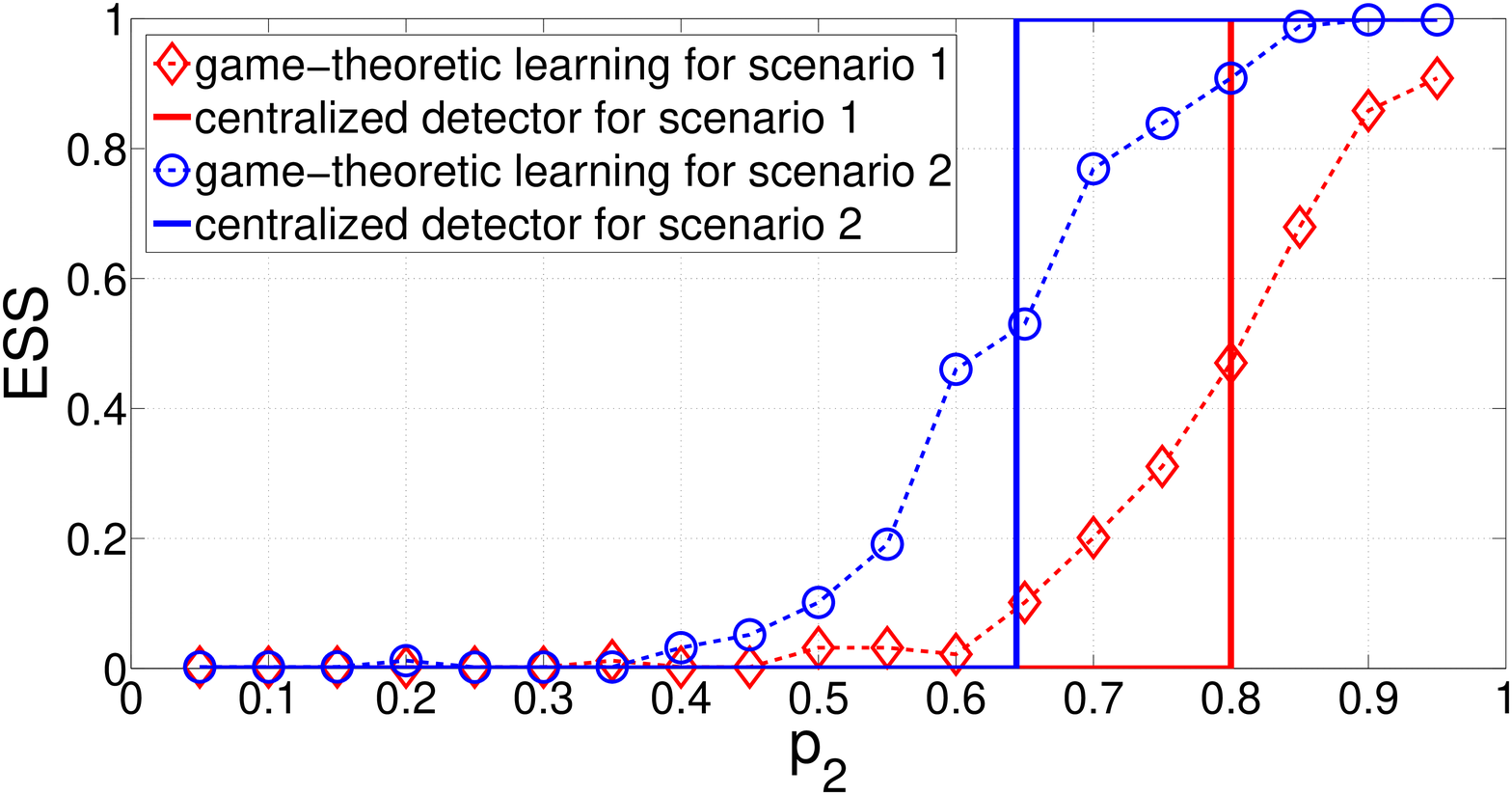}}
\subfigure[Five types of agents]{
\includegraphics[scale = 0.15]{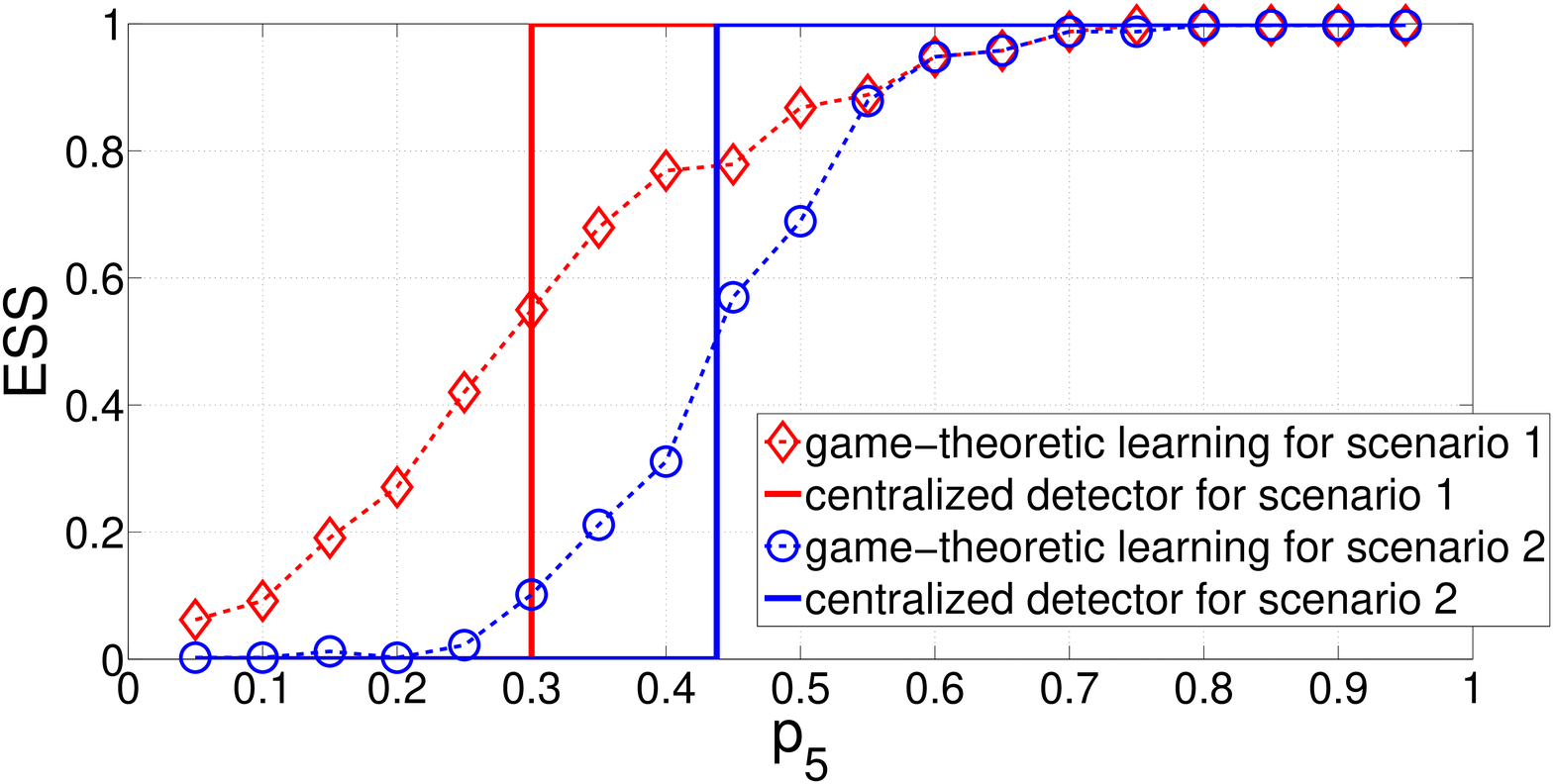}}
\caption{Performance of the proposed social learning method}
\label{simu}
\end{figure}

\section{Conclusion}
In this letter, a graphical evolutionary game based social learning method is proposed. The method is fully distributed and only requires communications of agents' binary decisions instead of their real-valued beliefs, which endows the proposed method with low communication complexity. Theoretical analysis under mean field approximations indicates that the evolutionarily stable states of the game coincide with the decisions of the centralized detector. Numerical experiments are implemented to validate the performance of the game-theoretic learning method.

\bibliography{mybib}{}
\bibliographystyle{ieeetr}

\end{document}